\documentclass[11pt,reqno]{amsart}
\usepackage{amsaddr}
\usepackage[left=2cm,right=2cm,top=2cm,bottom=2cm]{geometry}
\usepackage{amsmath,amsthm,amssymb,bm,bbm,dsfont}
\usepackage[mathscr]{eucal}
\usepackage{upgreek}
\usepackage{setspace}
\usepackage{graphicx}
\usepackage{verbatim}
\usepackage{float}
\usepackage{placeins}
\usepackage{array}
\usepackage{booktabs}
\usepackage{threeparttable}
\usepackage[update,prepend]{epstopdf}
\usepackage{multirow}
\usepackage{amsfonts,amssymb,dsfont}
\usepackage[abs]{overpic}

\usepackage[usenames,dvipsnames]{color}
\usepackage[hidelinks]{hyperref}
\hypersetup{
	unicode=false,          
	pdftoolbar=true,        
	pdfmenubar=true,        
	pdffitwindow=false,     
	pdfstartview={FitH},    
	pdftitle={My title},    
	pdfauthor={Author},     
	pdfsubject={Subject},   
	pdfcreator={Creator},   
	pdfproducer={Producer}, 
	pdfkeywords={keyword1} {key2} {key3}, 
	pdfnewwindow=true,      
	colorlinks=true,        
	linkcolor=Red,          
	citecolor=ForestGreen,  
	filecolor=Magenta,      
	urlcolor=BlueViolet,    
}
\usepackage{doi}
\usepackage{url}
\usepackage{caption, subcaption}
\usepackage{enumitem}

\makeatletter
\ifx\@NODS\undefined%

\let\mathbb=\mathds
\else%
\fi
\makeatother

\DeclareMathOperator*{\argmax}{\arg\max}
\DeclareMathOperator*{\argmin}{\arg\min}
\DeclareMathOperator{\Tr}{Tr}

\DeclareMathOperator{\Var}{Var}

\newcommand{\be}{{\mathbf e}}

\def\0{{\mathbf{0}}}
\def\1{{\mathbf{1}}}
\def\2{{\mathbf{2}}}
\def\3{{\mathbf{3}}}
\def\4{{\mathbf{4}}}
\def\5{{\mathbf{5}}}
\def\6{{\mathbf{6}}}

\def\7{{\mathbf{7}}}
\def\8{{\mathbf{8}}}
\def\9{{\mathbf{9}}}

\def\be{\begin{equation}}
\def\ee{\end{equation}}
\def\bea{\begin{eqnarray}}
\def\eea{\end{eqnarray}}

\theoremstyle{plain}
\newtheorem{theo}{Theorem} 
\newtheorem{lemm}[theo]{Lemma} 

\theoremstyle{definition}
\newtheorem{defn}[theo]{Definition} 

\theoremstyle{remark}

\let\origmaketitle\maketitle
\def\maketitle{
	\begingroup
	\def\uppercasenonmath##1{} 
	\let\MakeUppercase\relax 
	\origmaketitle
	\endgroup
}

\begin{document}

\title{\bfseries \Large{Sphere-Packing Bound for Symmetric Classical-Quantum Channels}}

\author{ {Hao-Chung Cheng$^{1,2}$, Min-Hsiu Hsieh$^2$, and Marco  Tomamichel$^{2,3}$}}
\address{\small  	
	$^{1}$Graduate Institute Communication Engineering, National Taiwan University, Taiwan (R.O.C.)\\
	$^{2}$Centre for Quantum Software and Information,\\
	Faculty of Engineering and Information Technology, University of Technology Sydney, Australia
\\
	$^3$School  of  Physics,  The University  of  Sydney, Australia
	}
\email{\href{mailto:F99942118@ntu.edu.tw}{F99942118@ntu.edu.tw}}
\email{\href{mailto:Min-Hsiu.Hsieh@uts.edu.au}{Min-Hsiu.Hsieh@uts.edu.au}}
\email{\href{mailto:marcotom.ch@gmail.com}{marcotom.ch@gmail.com}}

%


\maketitle

\begin{abstract}
We provide a sphere-packing lower bound for the optimal error probability in finite blocklengths when coding over a symmetric classical-quantum channel. 
Our result shows that the pre-factor can be significantly improved from the order of the subexponential to the polynomial. This established pre-factor is essentially optimal because it matches the best known random coding upper bound in the classical case.
Our approaches rely on a sharp concentration inequality in strong large deviation theory and crucial properties of the error-exponent function.
\end{abstract}

\section{Introduction} \label{sec:introduction}

The probability of decoding error is one of the fundamental criteria for evaluating the performance of a communication system.
In Shannon's seminal work \cite{Sha48}, he pioneered the study of  the noisy coding theorem, which states that the error probability can be made arbitrarily small as the coding blocklength grows when the coding rate $R$ is below the channel capacity $C$.
Later, Shannon \cite{Sha59} made a further step in exploring the exponential dependency of the optimal error probability $\epsilon^*(n,R)$ on the blocklength $n$ and rate $R$, and defined the \emph{reliability function} as follows: given a fixed coding rate $R<C$,
$E(R) := \limsup_{n\to+\infty} \, -\frac1n \log \epsilon^*(n,R)$.
The quantity $E(R)$ then provides a measure of how rapidly the error probability approaches zero with an increase in blocklength. This asymptotic characterization of the optimal error probability under a fixed rate is hence called the \emph{error exponent analysis}.
For a classical channel, the upper bounds of the optimal error can be established using a random coding argument \cite{Gal68}. 
On the other hand, the lower bound was first developed by Shannon, Gallager, and Berlekamp \cite{SGB67} and was called the \emph{sphere-packing bound}.
Alternative approaches by Haroutunian \cite{Har68} and Blahut \cite{Bla74} were subsequently proposed. 

In recent years, much attention has been paid to the finite blocklength regime \cite{PPV10,TBR16}. Altu\u{g} and Wagner employed strong large deviation techniques \cite{BR60} to prove a sphere-packing bound with a finite blocklength $n$. Moreover, the pre-factor of the bound was significantly refined from the order of the subexponential $\exp\{ -O(\sqrt{n})  \}$ \cite{SGB67} to the polynomial \cite{AW11,AW14}.
This refinement is substantial especially at rates near capacity, where the error-exponent function is close to zero; hence, the pre-factor dominants the bound \cite{AW10, AW14b}.

Error exponent analysis in classical-quantum (c-q) channels is much more difficult because of the noncommutative nature of quantum mechanics. Burnashev and Holevo \cite{BH98, Hol00} investigated reliability functions in c-q channels and proved the random coding upper bound for pure-state channels.
Winter \cite{Win99} adopted Haroutunian's method to derive a sphere-packing bound for c-q channels in the form of relative entropy functions \cite{Har68}.  
Dalai \cite{Dal13} employed Shannon-Gallager-Berlekamp's approach to establish a sphere-packing bound with Gallager's expression \cite{SGB67}.
It was later pointed out that these two sphere-packing exponents are not equal for general c-q channels \cite{DW14}.
In this work, we initiate the study of the refined sphere-packing bound in the quantum scenario.
In particular, we consider a ``symmetric c-q channel" (see Section~\ref{sec:notation} for a detailed definition), which is an important class of \emph{covariant channels} (e.g.~\cite{Hol93}), and establish a sphere-packing bound with the pre-factor improved
from the order of the subexponential in Dalai's result \cite{Dal13} to the polynomial.
Our result recovers Altu\u{g} and Wagner's work \cite{AW11} for classical symmetric channels including the binary symmetric channel and binary erasure channel.
Furthermore, the proved pre-factor matches that of the best known random coding upper bound \cite{Hon15} in the classical case.
Hence, our result yields the exact asymptotics for the sphere-packing bound in symmetric c-q channels.
The main ingredients in our proof are a tight concentration inequality from Bahadur and Ranga Rao \cite{BR60}, \cite{AW14b} (see Appendix \ref{app:tight}) and the major properties of the sphere-packing exponent \cite{HM16}.
We remark that the result obtained in this paper might enable analysis in the medium error probability regime of a classical-quantum channel \cite{AW10,AW14b,CCT+16a}. 
We leave the case for general c-q channels as future work \cite{CHT16b}.

This paper is organized as follows. We introduce the necessary notation and state our main result in Section \ref{sec:notation}.
Section \ref{sec:Properties} includes the crucial properties of the error-exponent function. We provide the proof of the main result in Section \ref{sec:proof}. Section \ref{sec:conclusions} concludes this paper.

\section{Notation and Main Result} \label{sec:notation}

\subsection{Notation} \label{ssec:notation}

Throughout this paper, we consider a finite-dimensional Hilbert space $\mathcal{H}$. 
The set of density operators (i.e.~positive semi-definite operators with unit trace) on $\mathcal{H}$ are defined as $\mathcal{S(H)}$.
For $\rho,\sigma\in\mathcal{S(H)}$, we write $\rho\ll \sigma$ if $\texttt{supp}(\rho) \subset \texttt{supp}(\sigma)$, where $\texttt{supp}(\rho)$ denotes the support of $\rho$. The identity operator on $\mathcal{H}$ is denoted by $\mathds{1}_\mathcal{H}$. When there is no possibility of confusion, we skip the subscript $\mathcal{H}$.  We use $\Tr\left[\,\cdot\, \right]$ as the trace function. 
Let $\mathbb{N}$, $\mathbb{R}$, and $\mathbb{R}_{> 0}$ denote the set of integers, real numbers,  and positive real numbers,, respectively.
Define $[n] := \{1,2,\ldots, n\}$ for $n\in\mathbb{N}$.
Given a pair of positive semi-definite operators $\rho,\sigma\in\mathcal{S(H)}$, we define the {(quantum) relative entropy}   as
$\mathbb{D}(\rho\|\sigma) :=  \Tr \left[ \rho \left( {\log} \rho - {\log} \sigma \right) \right]$,
when $\rho\ll\sigma$, and $+\infty$ otherwise.
For every $\alpha\in [0,1)$, we define the (Petz) quantum R\'enyi divergences 
$D_\alpha(\rho\|\sigma) := \frac{1}{\alpha-1} \log \Tr \left[ \rho^\alpha \sigma^{1-\alpha} \right]$.
For $\alpha = 1$, $D_1(\rho\|\sigma) := \lim_{\alpha\to 1} D_\alpha(\rho\|\sigma) = \mathbb{D}(\rho\|\sigma)$.
Let $\mathcal{X} = \{1,2,\ldots, |\mathcal{X}| \}$ be a finite alphabet, and 
let $\mathscr{P}(\mathcal{X})$ be the set of probability distributions on $\mathcal{X}$. 
In particular, we denote by $U_\mathcal{X}$ the uniform distribution on $\mathcal{X}$.
A classical-quantum (c-q) channel $W$ maps elements of the finite set $\mathcal{X}$ to the density operators in $\mathcal{S}(\mathcal{H})$, i.e.,~$W:\mathcal{X}\to\mathcal{S}(\mathcal{H})$.
Let $\mathcal{M}$ be a finite alphabetical set with size $M=|\mathcal{M}|$. An ($n$-block) \emph{encoder} is a map $f_n:\mathcal{M}\to \mathcal{X}^n$ that encodes each message $m\in\mathcal{M}$ to a codeword $\mathbf{x}^n(m) :=   x_1(m) \ldots x_n(m) \in\mathcal{X}^n$.
The codeword $\mathbf{x}^n(m)$ is then mapped to a state
$W_{\mathbf{x}^n(m)}^{\otimes n} = W_{x_1(m)} \otimes \cdots \otimes W_{x_n(m)} \in \mathcal{S}(\mathcal{H}^{\otimes n})$.
The \emph{decoder} is described by a positive operator-valued measurement (POVM) $\Pi_n = \{\Pi_{n,1},\ldots, \Pi_{n,M} \}$ on $\mathcal{H}^{\otimes n}$, where $\Pi_{n,i} \geq 0$ and $\sum_{i=1}^{M} \Pi_{n,i} = \mathds{1}$. The pair $(f_n, \Pi_n) =: \mathcal{C}_n$ is called a \emph{code} with \emph{rate} $R = \frac1n \log |\mathcal{M}|$.  The error probability of  sending a message $m$ with the code $ \mathcal{C}_n$ is $\epsilon_m(W,\mathcal{C}_n) :=  1- \Tr\left(\Pi_{n,m} W_{\mathbf{x}^n(m)}\right)$. We use $\epsilon_\text{max}(W,\mathcal{C}_n) = \max_{m\in\mathcal{M}} \epsilon_m(W,\mathcal{C}_n) $ and $\bar{\epsilon}(W,\mathcal{C}_n) = \frac1M \sum_{m\in\mathcal{M}} \epsilon_m(W,\mathcal{C}_n)$ to denote the \emph{maximal} error probability and the \emph{average} error probability, respectively. 
Given a sequence ${\mathbf{x}}^n \in \mathcal{X}^n$, we denote by 
$P_{\mathbf{x}^n} (x) := \frac1n \sum_{i=1}^n \mathbf{1}\left\{ x = x_i \right\}$
the empirical distribution of $\mathbf{x}^n$.

Throughout this paper, we consider a \emph{symmetric c-q} channel defined as
\begin{align} \label{eq:sym}
W_x := V^{x-1} W_1 (V^{\dagger})^{{x-1}}, \quad \forall x\in\mathcal{X},
\end{align}
where $W_1\in\mathcal{S(H)}$ is an arbitrary density operator, and $V$ 
satisfies $V^{\dagger} V = V V^{\dagger} = V^{|\mathcal{X}|} = \mathds{1}_\mathcal{H}$.
We define the following conditional entropic quantities for the channel $W$ with $P\in\mathscr{P}(\mathcal{X})$:
${D}_\alpha \left( W \| \sigma | P \right) :=  \sum_{x\in\mathcal{X}} P(x) {D}_\alpha\left( W_x \| \sigma \right)$.
The \emph{mutual information} of the c-q channel $W: \mathcal{X}\to \mathcal{S(H)}$ with prior distribution $P\in\mathscr{P}(\mathcal{X})$ is defined as 
$I(P,W) := \mathbb{D} \left( W \| PW | P \right)$,
where $PW^\alpha := \sum_{x\in\mathcal{X}} P(x) W_x^\alpha $, $\alpha\in(0,1]$.
The (classical) \emph{capacity} of the channel $W: \mathcal{X}\to \mathcal{S(H)}$ is denoted by
$C := \max_{P\in\mathscr{P}(\mathcal{X})} I(P,W)$.
Let
\begin{align}
E_\text{sp}^{(1)} (R,P) &:=  \sup_{s\geq 0} \left\{ E_0(s,P) - sR  \right\} \notag \\
E_\text{sp}^{(2)} (R,P) &:= \sup_{0<\alpha\leq 1} \min_{\sigma\in\mathcal{S(H)}} 
\frac{\alpha-1}{\alpha}   \left( R - D_\alpha \left( W \|  \sigma | P \right) \right), \notag
\end{align}
where we denote by 
$E_0(s,P) :=  -\log \Tr \left[
\left(  P W^{1/(1+s)}\right)^{1+s}
\right]$
an {auxiliary function} \cite{Hol00, HM16}.
The \emph{sphere-packing exponent} is defined  by
\begin{align}
E_\text{sp}(R) := \max_{P\in\mathscr{P}(\mathcal{X})} E_\text{sp}^{(1)} (R,P)
= \max_{P\in\mathscr{P}(\mathcal{X})} E_\text{sp}^{(2)} (R,P),
\end{align}
where the last equality follows from \cite[Proposition IV.2]{MO14b}.
Further, we define a rate \cite[p.~152]{CK11}, \cite{Dal13}:
\begin{align} \label{eq:R_inf1}
R_\infty &:= \lim_{s\to+\infty} \max_{P\in\mathscr{P}(\mathcal{X})} \min_{\sigma\in\mathcal{S(H)}} \sum_{x\in\mathcal{X}} P(x) D_{\frac{1}{1+s}}\left(W_x\|\sigma \right) \notag \\
&= \max_{P\in\mathscr{P}(\mathcal{X})} \min_{\sigma\in\mathcal{S(H)}} \sum_{x\in\mathcal{X}} P(x) \Tr\left[ W_x^0 \sigma\right].
\end{align}
It follows that $E_\text{sp}(R) = +\infty$ for any $R\leq R_\infty$ 
(see also \cite[p.~69]{SGB67} and \cite[Eq.~(5.8.5)]{Gal68}).

Consider a binary hypothesis whose null and alternative hypotheses are $\rho\in\mathcal{S(H)}$ and $\sigma\in\mathcal{S(H)}$, respectively. The \emph{type-I error} and \emph{type-II error} of the hypothesis testing, for an operator $0\leq Q\leq \mathds{1}$, are defined as
$\alpha\left(Q;\rho\right):= \Tr\left[ (\mathds{1}-Q) \rho \right]$, and 
$\beta\left(Q;\sigma\right):= \Tr\left[ Q \sigma \right]$.
There is a trade-off between these two errors. Thus, we can define the minimum type-I error, when the type-II error is below $\mu\in(0,1)$, as
\begin{align} \label{eq:alpha}
\widehat{\alpha}_{\mu}\left(\rho\|\sigma\right)
:= \min_{0\leq Q\leq \mathds{1} } \big\{ \alpha\left(Q;\rho\right) : \beta\left(Q;\sigma\right) \leq \mu  \big\}.
\end{align}

\subsection{Main Result} \label{ssec:main}

Let us now consider any symmetric c-q channel with capacity $C$.

\begin{theo}[Exact Sphere-Packing Bound] \label{theo:refined}
	For any rate $R\in[0,C)$, there exist an $N_0\in\mathbb{N}$ such that for all codes $\mathcal{C}_n$ of length $n\geq N_0$, we have
	\begin{align} \label{eq:goal}
	{\epsilon}_{\max} \left(\mathcal{C}_n \right) \geq  \frac{1-o(1)}{ n^{\frac12 \left( 1+\left| E_\textnormal{sp}'(R)\right| \right) } }   \exp\left\{
	-n E_\textnormal{sp}(R)
	\right\},	
	\end{align}
	where $E_\textnormal{sp}'(R) := \partial \max_{P\in\mathscr{P}(\mathcal{X})} E_\textnormal{sp}^{(1)}(r,P)/ \partial r |_{r=R}$.
\end{theo}

\section{Properties of the Sphere-Packing Exponent} \label{sec:Properties}


\begin{lemm}
	[Optimal Input Distribution] \label{lemm:input}
	For any $R> R_\infty$, the distribution $U_{\mathcal{X}}$ is a maximizer of $E_\textnormal{sp}^{(1)}(R,\cdot)$ and $E^{(2)}_\textnormal{sp}(R,\cdot)$.
\end{lemm}
\begin{proof}
	We first prove that $U_\mathcal{X}$ attains $\max_{P\in\mathscr{P}(\mathcal{X})} E_0(s,P)$.
	From Eq.~\eqref{eq:sym}, it is not hard to verify that $U_{\mathcal{X}} W^{\alpha} = V U_{\mathcal{X}} W^\alpha V^\dagger$ for all $\alpha
	\in (0,1]$.
	Hence,  it follows that 
\begin{align}
\Tr[ W_x^{\alpha}  ( U_{\mathcal{X}} W^{\alpha} )^{\frac{1-\alpha}{\alpha}} ]
&= \Tr[ V^{x-1} W_1^{\alpha} V^{\dagger\, x-1}  ( U_{\mathcal{X}} W^{\alpha} )^{\frac{1-\alpha}{\alpha}} ] \\
&= \Tr[ W_1^{\alpha} V^{\dagger\, x-1}  ( U_{\mathcal{X}} W^{\alpha} )^{\frac{1-\alpha}{\alpha}}  V^{x-1}] \\
&= \Tr[  W_1^{\alpha}  ( U_{\mathcal{X}} W^{\alpha} )^{\frac{1-\alpha}{\alpha}} ] \\
&=
\Tr[ ( U_{\mathcal{X}} W^{\alpha} )^{\frac{1}{\alpha}}]
\end{align} 
	for all $\alpha\in (0,1]$.
	The above equation shows that
	the distribution $U_\mathcal{X}$ that maximizes $E_0(s,P)$, $\forall s\geq 0$ \cite[Eq.~(38)]{Hol00}.
	Then we have
	\begin{align*}
	E_\text{sp}^{(1)}(R, U_\mathcal{X}) = \sup_{s\geq 0} \left\{ \max_{P\in\mathscr{P}(\mathcal{X})} E_0(s,P) -sR  \right\} = E_\text{sp}(R).
	\end{align*}
	Further, Jensen's inequality implies that $E_\text{sp}^{(2)} (R,U_\mathcal{X}) \geq E_\text{sp}^{(1)} (R,U_\mathcal{X}) = E_\text{sp}(R)$, which completes the proof.
\end{proof}

\begin{lemm}[Saddle-Point Property] \label{lemm:saddle}
	Consider any $R\in(R_\infty, C   )$ and $P\in \mathscr{P}(\mathcal{X})$.
	Let
	$\mathcal{S}_{P,W}(\mathcal{H}) := \left\{ \sigma \in \mathcal{S(H)}: \forall x \in \textnormal{\texttt{supp}}(P), \, \textnormal{\texttt{supp}}(W_x) \cap \textnormal{\texttt{supp}} (\sigma)  \neq \emptyset \right\}$. We define
	\begin{align} \label{eq:F}
	F_{R,P} (\alpha, \sigma) := \frac{\alpha-1}{\alpha} \left( R -  D_\alpha\left( W \| \sigma | P  \right)  \right), 
	\end{align}
	on $(0,1]\times \mathcal{S}_{P,W}(\mathcal{H})$, and let
	$	\mathscr{P}_R(\mathcal{X}) := \left\{ P\in\mathscr{P}(X) : \min_{\sigma \in \mathcal{S(H)}} \sup_{0<\alpha\leq 1} F_{R,P} (\alpha, \sigma) \in \mathbb{R}_{>0}    \right\}$.
	The following holds
	\begin{itemize}
		\item[(i)] For any $P\in\mathscr{P}(\mathcal{X})$, $F_{R,P}(\cdot,\cdot)$ has a saddle-point with  the saddle-value:
				\begin{align}
		\min_{\sigma \in \mathcal{S(H)}}\sup_{0<\alpha\leq 1}   F_{R,P} (\alpha,\sigma) 
		= \sup_{0<\alpha\leq 1} \min_{\sigma \in \mathcal{S(H)}}  F_{R,P} (\alpha,\sigma) 
		= E_\textnormal{sp}^{(2)}(R,W,P).
				\end{align}
		
		\item[(ii)] The saddle-point is unique for $P\in\mathscr{P}_R(\mathcal{X})$.
		
		\item[(iii)] Let $P\in\mathscr{P}_R(\mathcal{X})$. The unique saddle-point $(\alpha, \sigma)$ of $F_{R,P}(\cdot,\cdot)$ satisfies $\alpha\in(0,1)$ and 
		\begin{align}
		\sigma = \frac{\left( \sum_{x\in\mathcal{X}} P(x) W_x^{\alpha} \mathrm{e}^{(1-\alpha)D_{\alpha}\left( W_x\| \sigma \right)} \right)^{1/\alpha}}{\Tr\left[ \left( \sum_{x\in\mathcal{X}} P(x) W_x^{\alpha} \mathrm{e}^{(1-\alpha)D_{\alpha}\left( W_x\| \sigma \right)} \right)^{1/\alpha}\right] }
		\gg W_x, \quad \forall x \in \textnormal{\texttt{supp}}(P).
		\end{align}
	\end{itemize}
	
\end{lemm}
\noindent The proof is provided in Appendix \ref{proof:saddle}.

\begin{lemm}[Representation] \label{lemm:opt}
	For any $R\in(R_\infty, C   )$, let $(\alpha_R^\star, \sigma_R^\star)$ be the saddle-point of $F_{R,{U}_\mathcal{X}}(\cdot,\cdot)$.
	It follows that
	\begin{align} \label{eq:opt}
	\left(\alpha_R^\star,
	\sigma_R^\star \right) = 
	\left( -E_\textnormal{sp}'(R),
	\frac{ \left(U_\mathcal{X} W^{\alpha_R^\star}\right)^{1/\alpha_R^\star} }{ \Tr \left[ \left(U_\mathcal{X} W^{\alpha_R^\star}\right)^{1/\alpha_R^\star} \right] } \right).
	\end{align}
\end{lemm}
\begin{proof}
	Since Lemma \ref{lemm:input} implies that ${U}_\mathcal{X}$ attains $E_\text{sp}^{(2)} (R,\cdot)$, one observes from the definition of $E_\text{sp}^{(2)}$ that all the quantities $D_{\alpha_R^\star}(W_x\|\sigma_R^\star)$, $x\in\mathcal{X}$ are equal.
	By item (iii) of Lemma \ref{lemm:saddle}, we obtain a representation of $\sigma_R^\star$ in Eq.~\eqref{eq:opt}.
	The optimal $\alpha_R^\star = - \partial E_\textnormal{sp}(r,U_\mathcal{X}) / \partial r |_{r=R}$ follows from \cite[Eq.~(42)]{HM16}.
\end{proof}

\begin{lemm}
	[Invariance] \label{lemm:invariance}
	For any $R\in(R_\infty, C   )$, we have
	\begin{align}
	F_{R,P}( \alpha_R^\star, \sigma_R^\star) = E_\textnormal{sp}(R) > 0, \quad \forall P\in\mathscr{P}(\mathcal{X}),
	\end{align}
	where $\alpha_R^\star$ and $\sigma_R^\star$ are defined in Eq.~\eqref{eq:opt}.
\end{lemm}
\begin{proof}
	Following the argument in Lemma \ref{lemm:input} and recalling Eq.~\eqref{eq:opt} in Lemma \ref{lemm:opt}, one can verify that $\sup_{\alpha\in(0,1]} F_{R,P}(\alpha, \sigma_R^\star) = \sup_{s\geq 0} \left\{ E_0(s,U_\mathcal{X}) - sR  \right\} = E_\text{sp}(R)$ for all $P\in\mathscr{P}(\mathcal{X})$. 
	Further, we obtain $E_\text{sp}(R) > 0$ for $R\in(R_\infty, C   )$ from the result in  \cite[Proposition 10]{HM16}.
\end{proof}

\section{Proof of the Main Result} \label{sec:proof}

For rates in the range $R\leq R_\infty$, we have $E_\text{sp}(R) = +\infty$. The bound in Eq.~\eqref{eq:goal} obviously holds.
Hence, we consider the case of $ R\in(R_\infty,C)$ and fix the rate throughout the proof.

We first pose the channel coding problem into a binary hypothesis testing through Lemma \ref{lemm:hypothesis}, which originates from Blahut \cite{Bla74} for the classical case. 
\begin{lemm}[Hypothesis Testing Reduction] \label{lemm:hypothesis}
	For any code $\mathcal{C}_n$ with message size $\mathrm{e}^{nr}$, there exists an $\mathbf{x}^n \in \mathcal{C}_n$ such that
	\begin{align}
	\epsilon_{\max}\left(\mathcal{C}_n\right) \geq \max_{  \sigma \in \mathcal{S}(H) } \widehat{\alpha}_{\exp\{-nr\}} \left( W_{\mathbf{x}^n}^{\otimes n} \| \sigma^{\otimes n} \right). 
	\end{align}		
\end{lemm}
\noindent The proof is provided in Appendix \ref{proof:hypothesis}.

Let us now commence with the proof of Theorem \ref{theo:refined}.
Fix arbitrary $\gamma, \xi >0$. Let $\gamma_n := \left( \frac12 + \gamma\right) \frac{\log n}{n}$ and $R_n := R - \gamma_n$. The choice of the rate back-off term $\gamma_n$ will become evident later.
Choose $N_1 \in\mathbb{N}$ such that $R_n \geq R - \xi > R_\infty$.
Let $\sigma_R^\star$ be defined in Eq.~\eqref{eq:opt}, and from Lemma \ref{lemm:hypothesis}, we have 
\begin{align} \label{eq:sharp19}
\epsilon_{\max}\left(\mathcal{C}_n\right) \geq \widehat{\alpha}_{\exp\{-n R_n\}} \left( W_{\mathbf{x}^n}^{\otimes n} \| \sigma_R^{\star{\otimes n}} \right).
\end{align}

In the following, we provide a lower bound for the type-I error $ \widehat{\alpha}_{\exp\{-n R_n\}} \left( W_{\mathbf{x}^n}^{\otimes n} \| \sigma_R^{\star{\otimes n}} \right)$.
Let ${p}^n := \bigotimes_{i=1}^n p_{x_i}$ and ${q}^n := \bigotimes_{i=1}^n q_{x_i}$, where $(p_{x_i},q_{x_i})$ are Nussbaum-Szko{\l}a distributions \cite{NS09} of $(W_{x_i},\sigma_R^\star)$ for every $i\in [n]$. Since $D_\alpha(W_{x_i}\|\sigma_R^\star) = D_\alpha(p_{x_i}\|q_{x_i})$, for all  $\alpha\in(0,1]$,  we shorthand 
$\phi_n(R_n) := \sup_{\alpha\in(0,1]} F_{R_n,P_{\mathbf{x}^n}} (\alpha, \sigma_{R}^\star )$, where $P_{\mathbf{x}^n}$ is the empirical distribution of $\mathbf{x}^n$.
Moreover, item (iii) in Lemma \ref{lemm:saddle} implies that the state $\sigma_R^\star$ dominates all the channel outputs: $\sigma^\star_R \gg W_x$, for all $x\in \texttt{supp}(P_{\mathbf{x}^n})$, 
Hence, we have $p^n\ll q^n$. 
Subsequently, for every $i\in[n]$, we let $q_{x_i}(\omega) = 0$, for all $\omega \not\in \texttt{supp}(p_{x_i})$.
We apply Nagaoka's argument \cite{Nag06}  by choosing $\delta = \exp\{n R_n - n\phi_n(R_n)\}$ to yield, for any $0\leq Q_n\leq \mathds{1}$,
\begin{align}  \label{eq:sharp15}
&\alpha\left(Q_n; W_{\mathbf{x}^n}^{\otimes n} \right) + \delta \beta\left(Q_n; \sigma_R^{\star{\otimes n}} \right) \geq  \frac{ \alpha\left(\mathscr{U};p^n \right) + \delta \beta\left( \mathscr{U};q^n\right) }{2},
\end{align}
where
$\alpha\left( {\mathscr{U}};  {p}^n \right) := \sum_{\omega\in {\mathscr{U}}^\mathrm{c}}  {p}^n(\omega)$, 
$\beta\left( {\mathscr{U}}; {q}^n\right) := \sum_{\omega\in {\mathscr{U}}}  {q}^n(\omega)$,	
and  
${\mathscr{U}} := \left\{  \omega:  p^n(\omega)\mathrm{e}^{ n{\phi}_n\left( {R_n}\right)} >  q^n(\omega) \mathrm{e}^{ n{R_n}}  \right\}$.

Next, we employ Bahadur-Ranga Rao's concentration inequality, Theorem \ref{theo:Rao} in Appendix \ref{app:tight}, to further lower bound $\alpha\left( {\mathscr{U}};  {p}^n \right)$ and $\beta\left( {\mathscr{U}}; {q}^n\right)$.
Before proceeding, we need to introduce some notation.
We define the \emph{tilted distributions}, for every $i\in[n]$, $\omega \in \texttt{supp}(p_{x_i})$, and $t\in[0,1]$ by
\begin{align}
\hat{q}_{{x_i},t}(\omega) := \frac{  {p}_{x_i}(\omega)^{1-t}  {q}_{x_i}(\omega)^{t} }{ \sum_{\omega\in \texttt{supp}(p_{x_i})}  {p}_{x_i}(\omega)^{1-t}  {q}_{x_i}(\omega)^{t} }.
\end{align}
Let
\begin{align}
\begin{split}
\Lambda_{0,{x_i}} (t) &:= \log \mathbb{E}_{ {p}_{x_i}} \left[ \mathrm{e}^{t \log \frac{ {q}_{x_i}}{ {p}_{x_i}} } \right]; \\
\Lambda_{1,{x_i}} (t) &:= \log \mathbb{E}_{ {q}_{x_i}} \left[ \mathrm{e}^{t \log \frac{ {p}_{x_i}}{ {q}_{x_i}} } \right].
\end{split}
\label{eq:zero_derivative}	
\end{align}
Since $p^n$ and $q^n$ are mutually absolutely continuous, the maps $t\mapsto \Lambda_{j,x_i}(t)$, $j\in\{0,1\}$ are differentiable  for all $t\in[0,1]$. 
One can immediately verify the following partial derivatives with respect to $t$:
\begin{align}
\begin{split}
&\Lambda'_{0,{x_i}} (t) = \mathbb{E}_{\hat{q}_{{x_i},t}} \left[ \log \frac{ {q}_{x_i}}{ {p}_{x_i}}  \right], \;  \Lambda''_{0,{x_i}} (t) = \Var_{\hat{q}_{{x_i},t}} \left[ \log \frac{ {q}_{x_i}}{ {p}_{x_i}}  \right],\\
&\Lambda''_{0,{x_i}} (t) = \Var_{\hat{q}_{{x_i},t}} \left[ \log \frac{ {q}_{x_i}}{ {p}_{x_i}}  \right], \;
\Lambda'_{1,{x_i}} (t)  = \mathbb{E}_{\hat{q}_{{x_i},1-t}} \left[ \log \frac{ {p}_{x_i}}{ {q}_{x_i}}  \right].
\label{eq:second_derivative} 
\end{split}
\end{align}
With $\Lambda_{j,{x_i}} (t)$ in Eq.~\eqref{eq:zero_derivative}, we can define 
\begin{align}
&{\Lambda}_{j,P_{\mathbf{x}^n}} (t) := \sum_{x\in \mathcal{X} } P_{\mathbf{x}^n} (x) \Lambda_{j,x}(t), \quad\quad\;
j\in\{0,1\}; \label{eq:FL0}\\
&\Lambda_{j,P_{\mathbf{x}^n}}^*(z) := \sup_{t\in\mathbb{R}} \left\{ tz - {\Lambda}_{j,P_{\mathbf{x}^n}}(t)  \right\},
\quad j\in\{0,1\}, \label{eq:FL}
\end{align}
where $\Lambda_{j,P_{\mathbf{x}^n} }^*(z)$ in Eq.~(\ref{eq:FL}) is the \emph{Fenchel-Legendre transform} of ${\Lambda}_{j,P_{\mathbf{x}^n}}(t)$. The quantities $\Lambda_{j,P_{\mathbf{x}^n} }^*(z)$ would appear in the lower bounds of $\alpha\left( {\mathscr{U}};  {p}^n \right)$ and $\beta\left( {\mathscr{U}}; {q}^n\right)$ obtained by Bahadur-Randga Rao's inequality as shown later.

In the following, we relate the Fenchel-Legendre transform $\Lambda_{j,P_{\mathbf{x}^n}}^* (z)$ to the desired error-exponent function $\phi_n(R_n)$.
Such a relationship is stated in Lemma \ref{lemm:regularity}; the proof is provided in Appendix \ref{proof:regularity}.
\begin{lemm} \label{lemm:regularity}
	Under the prevailing assumptions and for all $R_n\in(R_\infty,C)$, the following holds:
	\begin{itemize}
		
		\item[\textnormal{(i)}] $\Lambda^*_{0,P_{\mathbf{x}^n} } \left(  {\phi}_n({R_n}) - {R_n} \right) =  {\phi}_n(R_n)$;
		\item[\textnormal{(ii)}] $\Lambda^*_{1,P_{\mathbf{x}^n} } \left( R_n -  {\phi}_n({R_n}) \right) = R_n$;
		\item[\textnormal{(iii)}] There exists a unique $t^\star = \frac{s^\star}{1+s^\star} \in (0,1)$, such that ${\Lambda}'_{0,P_{\mathbf{x}^n} }(t^\star) =  {\phi}_n(R_n) - R_n$, where $s^\star := \frac{ \partial  {\phi}_n (r) }{\partial r}|_{r=R_n}$.
	\end{itemize}
\end{lemm}
Item (iii) in Lemma \ref{lemm:regularity} shows that the optimizer $t$ in Eq.~\eqref{eq:FL} always lies in the compact set $[0,1]$. Further, Eqs.~\eqref{eq:zero_derivative} and \eqref{eq:second_derivative} ensure that 
$\Lambda_{0,x_i}(t) = \Lambda_{1,x_i} (1-t)$, 
$\Lambda_{0,x_i}'(t) = -\Lambda_{1,x_i}' (1-t)$, 
$\Lambda_{0,x_i}''(t) = \Lambda_{1,x_i}'' (1-t).$ 
We define the following quantities:
\begin{align}
V_{\max}  &:= \max_{t \in [0,1],\, x\in\mathcal{X} }  {\Lambda}''_{0,x}(t);  \\
V_{\min} &:= \min_{ t \in [0,1],\, x\in\mathcal{X} }  {\Lambda}''_{0,x}(t);\label{eq_mhvmin} \\
T_{\max}  &:= \max_{t \in [0,1],\, x\in\mathcal{X} }  {T}_{0,x}(t); \\
T_{0, x }(t) &:= \mathbb{E}_{\hat{q}_{x,t}} \left[ \left| \log \frac{ {q}_x}{ {p}_x} - \Lambda'_{0,x}(t) \right|^3 \right]; \label{eq:third_derivatie}
\end{align}
$T_{1,x}(t) := T_{0,x}(1-t)$; and $K_{\max} := 15\sqrt{2\pi} T_{\max}/V_{\min}$.
Note that for every $x\in\mathcal{X}$, $\Lambda_{0,x}''(\cdot)$ and $T_{0,x}(\cdot)$ are continuous functions on $[0,1]$ from the definitions in Eqs.~\eqref{eq:second_derivative}, \eqref{eq:third_derivatie} (see also \cite[Lemma 9]{AW11}).
The maximization and minimization in the above definitions are well-defined and finite.
Moreover, Lemma \ref{lemm:positivity} guarantees that
$V_{\min}$ is bounded away from zero.
\begin{lemm}
	[Positivity] \label{lemm:positivity}
	For any $R_n\in(R_\infty,C)$ and $P_{\mathbf{x}^n} \in\mathscr{P}(\mathcal{X})$, $\Lambda_{0,P_\mathbf{x}^n}''(t) > 0$, for all $t\in[0,1]$.
\end{lemm}
\begin{proof}
	Assume $\Lambda_{0,P_\mathbf{x}^n }''(t)$ is zero for some $t\in[0,1]$. This is equivalent to
	\begin{align} \label{eq:sharp3}
	{p}_{x_i} (\omega) =  {q}_{x_i} (\omega) \cdot \mathrm{e}^{-\Lambda'_{0,x_i}(t)}, 
	\quad \forall \omega \in p_{x_i}, 
	\quad \forall i\in [n]. 
	\end{align}	
	Summing the right-hand side of Eq.~\eqref{eq:sharp3} over $\omega \in p_{x_i}$ gives
	$1 = \Tr\left[ p_{x_i}^0 q_{x_i} \right]  \mathrm{e}^{-\Lambda'_{0,{x_i}}(t)}, \quad \forall i\in [n]$.
	Then, Eqs.~\eqref{eq:sharp3} 
	and the above equation imply that
	\begin{align}
	\phi_n ( R_n ) 
	&= \sup_{0<\alpha\leq 1} \frac{\alpha-1}{\alpha} \left( R_n + \sum_{x\in\mathcal{X}} P_{\mathbf{x}^n} (x)\log \Tr\left[ p_x^0 q_x \right]    \right) \notag\\
	&= 0, \notag
	\end{align}
	where  we use the fact that  $R_n > R_\infty = -\sum_{x\in\mathcal{X}} P_{\mathbf{x}^n} (x)\log \Tr\left[ p_x^0 q_x \right]$; see Eq.~\eqref{eq:R_inf1}).
	However, Lemma \ref{lemm:invariance} implies that $\phi_n(R_n) =  E_\text{sp}(R_n) >0$, 
	which leads to a contradiction.
\end{proof}


Now, we are ready to derive the lower bounds to $\alpha\left({\mathscr{U}};  p^n \right)$ and $\beta\left({\mathscr{U}}; q^n \right)$.
Let $N_2\in\mathbb{N}$ be sufficiently large such that for all $n\geq N_2$,
\begin{align} \label{eq_largen}
\sqrt{n} \geq 
\frac{ 1+ \left( 1 + K_{\max}  \right)^2 }{ \sqrt{ V_{\min}   } }.
\end{align}
Applying Bahadur-Randga Rao's inequality (Theorem \ref{theo:Rao}) to  $Z_i = \log  {q}_i - \log  {p}_i$ with the probability measure $\lambda_i =  {p}_i$, and $z =  {R_n} -  {\phi}_n( {R_n}) $ gives
\begin{align}
\alpha\left( {\mathscr{U}};  {p}^n \right)
&= \Pr\left\{ \frac1n\sum_{i=1}^n Z_i \geq  R_n - {\phi}_n( {R_n})  \right\} \\	
&\geq \frac{2A}{\sqrt{n}} \exp\left\{  -n \Lambda^*_{0,P_{\mathbf{x}^n} } \left(  {\phi}_n( {R_n}) -  {R_n} \right)  \right\} \label{eq:sharp12mh}
\end{align}
where
$A  := 
\frac{ \mathrm{e}^{-K_{\max} } }{ \sqrt{ 4\pi V_{\max}  }    }.$
Similarly, applying Theorem \ref{theo:Rao} to $Z_i = \log  {p}_i - \log  {q}_i$ with the probability measure $\lambda_i =  {q}_i$, and $z =  {\phi}_n( {R_n}) -  {R_n}$ yields
\begin{align}
\beta\left( {\mathscr{U}}; {q}^n\right) 
&= \Pr\left\{ \frac1n\sum_{i=1}^n Z_i \geq  {\phi}_n( {R_n}) -  {R_n}  \right\} \\
&\geq \frac{2A}{\sqrt{n}} \exp\left\{  -n \Lambda^*_{1,P_{\mathbf{x}^n} } \left(  {R_n} -  {\phi}_n( {R_n}) \right)  \right\}  \label{eq:sharp13mh}.
\end{align}

Continuing from Eq.~(\ref{eq:sharp12mh}) and item (i) in Lemma \ref{lemm:regularity} gives
\begin{equation}
\alpha\left( {\mathscr{U}};  {p}^n \right)\geq \frac{2A}{\sqrt{n}}  \exp\{-n   {\phi}_n\left( {R_n}\right)\}. \label{eq:sharp12}
\end{equation}
Eq.~(\ref{eq:sharp13mh}) together with item (iii) in Lemma \ref{lemm:regularity} yields
\begin{equation}
\beta\left( {\mathscr{U}}; {q}^n\right) \geq \frac{2A}{\sqrt{n}}  \exp\{-n {R_n}\} 
= 2A n^\gamma \exp\{-n {R}\}.
\label{eq:sharp13}
\end{equation}
Let $N_3\in\mathbb{N}$ such that $A n^\gamma > 1$, for all $n\geq N_3$. Then Eq.~\eqref{eq:sharp13} implies that $\beta\left( {\mathscr{U}}; {q}^n\right) > 2\exp\{-n {R}\}$.
Thus, we can bound the left-hand side of Eq.~(\ref{eq:sharp15}) from below by 
$\frac{A}{\sqrt{n}} \mathrm{e}^{-n \phi_n(R_n)}.$ For any test $0\leq Q_n \leq \mathbb{1}$ such that 
$\beta(Q_n;\sigma_R^{\star \otimes n}) \leq \exp\{-n {R}\}$,
we have
\begin{align}
&\widehat{\alpha}_{\exp\{-n R_n \}} \left( W_{\mathbf{x}^n}^{\otimes n} \| \sigma_R^{\star \otimes n} \right)  = \alpha(Q_n;\rho^n)  \notag \\
&\geq \frac{A}{\sqrt{n}} \exp\{-n   {\phi}_n\left( { R_n }\right)\} 
= \frac{A}{\sqrt{n}} \exp\left\{ -n E_\text{sp}(R_n)
\right\}, \label{eq:sharp20}
\end{align}
where the last equality follows from Lemma \ref{lemm:invariance}.

Finally, it remains to remove the back-off term $R_n = R - \gamma_n$ in Eq.~\eqref{eq:sharp20}.
By Taylor's theorem, we have
\begin{align} \label{eq:sharp24}
E_\text{sp}(R-\gamma_n) = E_\text{sp}(R) - \gamma_n  E_\text{sp}'(R) + \frac{\gamma_n^2}{2} E_\text{sp}''(\bar{R}),
\end{align}
for some $\bar{R} \in (R-\xi, R)$ and $E_\text{sp}''(\bar{R}) := \left.\frac{ \partial^2 E^{(1)}_\text{sp} (r, U_\mathcal{X})}{\partial r^2}\right|_{r = \bar{R}}$.
Further, one can calculate that
\begin{align}
E_\text{sp}''(\bar{R}) 
&= - \left.\left( \frac{\partial^2 E_0(s,U_\mathcal{X})}{\partial s^2}\right|_{s=\bar{s}} \right)^{-1} \\
&= \frac{ (1 + \bar{s})^3}{ \Lambda_{0,U_\mathcal{X}}''\left( \frac{\bar{s}}{1+\bar{s}}  \right) }
\leq \frac{ (1 + \bar{s})^3}{ V_{\min}} =: \Upsilon , \label{eq:sharp21}
\end{align}
where $\bar{s} = \frac{1-\alpha_{\bar{R}}^\star}{\alpha_{\bar{R}}^\star}$.
From item (iii) in Lemma \ref{lemm:saddle}, it follows that both $\bar{s}$ and $|E_\text{sp}'(R)| = s^\star$ are both positive and finite for $\bar{R}\in(R_\infty,C)$ and ${R}\in(R_\infty,C)$. Together with the fact that $V_{\min} > 0$, we have $ \Upsilon \in \mathbb{R}_{>0}$.
We apply Taylor's expansion on the function $n^{-(\cdot)}$ again to yield
\begin{align} 
n^{-\frac12 \left( 1 + \left|E_\text{sp}'(R)\right| \right) - \gamma_n \Upsilon } 
&= n^{-\frac12 \left( 1 + \left|E_\text{sp}'(R)\right| \right)} \cdot
\left( 1 - \frac{\log n}{n^{\bar{x}\Gamma}} \gamma_n \Upsilon \right) \notag \\
&= n^{-\frac12 \left( 1 + \left|E_\text{sp}'(R)\right| \right)} \cdot
\left( 1 - o(1) \right), \label{eq:sharp23}
\end{align}
where the first equality holds for some $\bar{x} \in (0, \gamma_n)$, and the last line follows from the definition $\gamma_n = (\frac12+\gamma) \frac{\log}{n}$.
Finally, by combining Eqs.~\eqref{eq:sharp19}, \eqref{eq:sharp20}, and \eqref{eq:sharp23}, we obtain the desired Eq.~\eqref{eq:goal}
for sufficiently large $n\geq N_0 := \max\left\{N_1, N_2, N_3\right\}$.

\section{Discussion} \label{sec:conclusions}

In this work, we establish a sphere-packing bound with a refined polynomial pre-factor that coincides with the best classical results \cite[Theorem 1]{AW11} to date.
As discussed by Altu\u{g} and Wagner \cite[Sec.~VII]{AW11}, the pre-factor is correct for binary symmetric channels but slightly worse for binary erasure channels (in the order of $1/\sqrt{n}$).
On the other hand, our pre-factor matches the recent result of the random coding upper bound \cite[Theorem 2]{Hon15}, where the pre-factor has been shown to be exact.
Hence, we conjecture that the established result is optimal for general symmetric c-q channels.

This work admits variety of potential extensions.
First, the symmetric c-q channel studied in this paper is a {covariant channel} with a cyclic group:
\begin{align}
W_{ \mathcal{U}_\text{in} (g) x \mathcal{U}_\text{in}(g)^\dagger  } = \mathcal{U}_\text{out}(g)  W_x \mathcal{U}_\text{out}(g)^\dagger, \quad \forall g,x\in \mathcal{X},
\end{align}
where $\mathcal{U}_\text{in}$ and $\mathcal{U}_\text{out}$ are the unitary representations on $\mathcal{X}$ and $\mathcal{S(H)}$ such that $\mathcal{U}_\text{in} (g)\, x \,\mathcal{U}_\text{in}(g)^\dagger = (x + g) \text{ mod } |\mathcal{X}|$ and $\mathcal{U}_\text{out}(g) = V^{g}$.
It would be interesting to investigate whether the refined sphere-packing bound can be extended to covariant quantum channels $\mathcal{N}: \mathcal{S}(\mathcal{H}_\text{in}) \to \mathcal{S}(\mathcal{H}_\text{out})$ with arbitrary compact groups.
Second, the random coding bound in the quantum case has been proved only for pure-state channels \cite{Hol00}. It is promising to prove the bound for this class of c-q channels by employing the symmetry property.
Finally, the refinement provides a new possibility for moderate deviation analysis in c-q channels \cite{AW14b}, which is left as future work.
\section*{Acknowledgements}
MH is supported by an ARC Future Fellowship under Grant FT140100574. 
MT is funded by an ARC Discovery Early Career Researcher Award (DECRA)
fellowship and acknowledges support from the ARC Centre of Excellence
for Engineered Quantum Systems (EQUS).
\appendix
\section{A Tight Concentration Inequality} \label{app:tight}
Let $\left(Z_i\right)_{i=1}^n$ be a sequence of independent, real-valued random variables whose probability measures are $\lambda_i$. Let $\Lambda_i(t) := \log \mathbb{E}\left[\mathrm{e}^{t Z_i}\right]$ and define the Fenchel-Legendre transform of $\frac1n\sum_{i=1}^n \Lambda_i(\cdot)$ to be:
$\Lambda_n^*(z) := \sup_{t \in \mathbb{R}} \left\{ zt - \frac1n \sum_{i=1}^n \Lambda_i(t) \right\}, \quad \forall z\in\mathbb{R}$.
Thenb there exists a real number $t^\star\in(0,1]$ for every $z\in\mathbb{R}$ such that 
$z=\frac1n\sum_{i=1}^n \Lambda_i'(t^\star)$ and 
$\Lambda_n^*(z) = zt^\star - \frac1n\sum_{i=1}^n \Lambda_i (t^\star)$. 
\label{eq:Rao1}
Define the probability measure $\tilde{\lambda}_i$ via 
$\frac{\mathrm{d}\tilde{\lambda}_i}{\mathrm{d}\lambda_i} (z_i) := \mathrm{e}^{t^\star z_i - \Lambda_i(t^\star)}$,
and let $\bar{Z}_i:= Z_i - \mathbb{E}_{\tilde{\lambda}_i}\left[Z_i\right]$. Furthermore, define  $m_{2,n} := \sum_{i=1}^n \Var_{\tilde{\lambda}_i}\left[\bar{Z}_i\right]$, $m_{3,n} := \sum_{i=1}^n \mathbb{E}_{\tilde{\lambda}_i}\left[\left|\bar{Z}_i\right|^3\right]$,	and $K_n(t^\star) := \frac{15\sqrt{2\pi}m_{3,n}}{m_{2,n}}$. With these definitions, we can now state the following sharp concentration inequality for $\frac1n\sum_{i=1}^n Z_i$:
\begin{theo}[Bahadur-Ranga Rao's Concentration Inequality {\cite[Proposition 5]{AW14}}, \cite{DZ98}] \label{theo:Rao} 
	Given
	\begin{align} \label{eq:m2n}
	\sqrt{m_{2,n}} \geq 1 + \left(1+K_n\left( t^\star \right) \right)^2,
	\end{align}
	it follows that
	\begin{align}
	\Pr\left\{ \frac1n\sum_{i=1}^n Z_i \geq z   \right\}
	\geq \mathrm{e}^{-n\Lambda_n^*(z)} \frac{ \mathrm{e}^{-K_n(t^\star)} }{ 2\sqrt{2\pi m_{2,n} }  }.
	\end{align}
\end{theo}

\section{Proofs of Miscellaneous Lemmas} \label{app:proofs}

\subsection{Proof of Lemma \ref{lemm:saddle}} \label{proof:saddle}
Let $R>R_\infty$ and $P\in\mathscr{P}(\mathcal{X})$ be arbitrary.
It is convenient to reparameterize the function $F_{R,P}$ by the substitution $\alpha = \frac{1}{1+s}$:
\begin{align}
\left. F_{R,P}\left( \alpha, \sigma \right)\right|_{\alpha = \frac{1}{1+s}} 
= - sR + s D_{\frac{1}{1+s}}\left( W \| \sigma | P\right)
=: K_{R,P} (s,\sigma). \label{eq:K}
\end{align}
In the following, we prove the existence of a saddle-point of $K_{R,P}(\cdot,\cdot)$ on $\mathbb{R}_{\geq 0}\times \mathcal{S}_{P,W}(\mathcal{H})$, where $\mathbb{R}_{\geq 0} := [0,\infty)$.
By Ref.~\cite[Lemma 36.2]{Roc64},
$(s^\star,\sigma^\star)$ is a saddle point of $K_{R,P}(\cdot,\cdot)$ if and only if the supremum in 
\begin{align}
\sup_{s\in\mathbb{R}_{\geq 0}} \inf_{\sigma \in \mathcal{S}_{P,W}(\mathcal{H})} K_{R,P} (s,\sigma) \label{eq:saddle15}
\end{align}
is attained at $s^\star$, the infimum in 
\begin{align}
\inf_{\sigma \in \mathcal{S}_{P,W}(\mathcal{H})} \sup_{s\in\mathbb{R}_{\geq 0}}  K_{R,P} (s,\sigma) \label{eq:saddle16}
\end{align}
is attained at $\sigma^\star$, and the two extrema in Eqs.~\eqref{eq:saddle15}, \eqref{eq:saddle16} are equal and finite.
We first claim that
\begin{align}
\forall s\in\mathbb{R}_{\geq 0}, \quad \inf_{\sigma \in \mathcal{S}_{P,W}(\mathcal{H})} K_{R,P}(s,\sigma)
= \inf_{\sigma \in \mathcal{S}(\mathcal{H})} K_{R,P}(s,\sigma). \label{eq:saddle21}
\end{align}
To see that, observe that for any $s \in \mathbb{R}_{\geq 0}$, the definition of the $\alpha$-R\'enyi divergence yields 
\begin{align}
\forall \sigma \in \mathcal{S(H)}\backslash \mathcal{S}_{P,W}(\mathcal{H}), \quad 
D_{\frac{1}{1+s}}\left( W\|\sigma | P \right) = +\infty, \label{eq:saddle23}
\end{align}
which, in turn, implies
\begin{align} \label{eq:saddle22}
\forall \sigma \in \mathcal{S(H)}\backslash \mathcal{S}_{P,W}(\mathcal{H}), \quad 
K_{R,P}(s,\sigma ) = +\infty.
\end{align}
Hence, Eq.~\eqref{eq:saddle21} yields
\begin{align}
\begin{split}
\sup_{s\in\mathbb{R}_{\geq 0}} \inf_{\sigma \in \mathcal{S}_{P,W}(\mathcal{H})} K_{R,P} (s,\sigma)
&= \sup_{s\in\mathbb{R}_{\geq 0}} \inf_{\sigma \in \mathcal{S}(\mathcal{H})} K_{R,P} (s,\sigma) 
=\sup_{s\in\mathbb{R}_{\geq 0}} \min_{\sigma \in \mathcal{S}(\mathcal{H})} K_{R,P} (s,\sigma),
\label{eq:saddle27}
\end{split}
\end{align}
where the last equality in Eq.~\eqref{eq:saddle27} follows from the
lower semi-continuity of the map $\sigma\mapsto D_{1/(1+s)}(W\|\sigma|P)$ \cite[Corollary III.25]{MO14b} and the compactness of $\mathcal{S(H)}$.
Further, by the fact $R> R_\infty$ and the definition of $E_\text{sp}^{(2)}$, we have
\begin{align}
E_\text{sp}^{(2)} (R,P) = \sup_{s\in\mathbb{R}_{\geq 0}} \min_{\sigma \in \mathcal{S}(\mathcal{H})} K_{R,P} (s,\sigma) < +\infty,
\end{align}
which guarantees the supremum in the right-hand side of Eq.~\eqref{eq:saddle27} is attained at some $s\in\mathbb{R}_{\geq 0}$, i.e.,
\begin{align} \label{eq:fact1}
\sup_{s\in\mathbb{R}_{\geq 0}} \inf_{\sigma \in \mathcal{S}_{P,W}(\mathcal{H})} K_{R,P} (s,\sigma)
= \max_{s\in\mathbb{R}_{\geq 0}} \min_{\sigma \in \mathcal{S}(\mathcal{H})} K_{R,P} (s,\sigma)  < +\infty.
\end{align}
Thus, we complete our claim in Eq.~\eqref{eq:saddle15}.
It remains to show that the infimum in Eq.\eqref{eq:saddle16} is attained at some $\sigma^\star \in \mathcal{S}_{P,W}(\mathcal{H})$ and the supremum and infimum are exchangeable.
To achieve this, we will show that $\left( \mathbb{R}_{\geq 0}, \mathcal{S}_{P,W}(\mathcal{H}), K_{R,P} \right)$ is a closed saddle-element (see Definition \ref{defn:saddle} below) and apply Rockafellar's saddle-point result, Theorem \ref{theo:saddle}, to conclude our claim.

\begin{defn} 
	[Closed Saddle-Element {\cite{Roc64}}] \label{defn:saddle}
	The triple $\left(\mathcal{A},\mathcal{B}, F\right)$ is called a closed saddle-element if for any\footnote{We denote by $\texttt{ri}$ and $\texttt{cl}$ the relative interior and the closure of a set, respectively.} 
	$x\in \texttt{ri}\left(\mathcal{A}\right)$ (resp.~$y\in \texttt{ri}\left(\mathcal{B}\right)$):
	\begin{itemize}
		\item[(a)] $\mathcal{B}$ (resp.~$\mathcal{A}$) is convex;
		\item[(b)] $F(x,\cdot)$ (resp.~$F(\cdot, y)$) is convex (resp.~concave) and lower (resp.~upper) semi-continuous; and
		\item[(c)] any accumulation point of $\mathcal{B}$ (resp.~$\mathcal{A}$) that does not belong to $\mathcal{B}$ (resp.~$\mathcal{A}$), say $y_o$ (resp.~$x_o$)
		satisfies $\lim_{y\to y_o} F(x,y) = +\infty$ (resp.~$\lim_{x\to x_o} F(x, y) = -\infty$).
	\end{itemize}
\end{defn}



\begin{theo}[The Existence of Saddle-Points {\cite[Theorem 8]{Roc64}, \cite[Theorem 37.3]{Roc70}}] \label{theo:saddle}
	Let $\left(\mathcal{A},\mathcal{B}, F\right)$ be any closed saddle-element on $\mathbb{R}^m\times \mathbb{R}^n$.
	\begin{itemize}
		\item[(I)] No non-zero $x_0$ has the property that, for all $x\in\textnormal{\texttt{ri}}(\mathcal{A})$ and $y\in\textnormal{\texttt{ri}}(\mathcal{B})$, the half-line $\left\{ x + t x_0: t\geq 0 \right\}$ is contained in $\mathcal{A}$ and $F(x+tx_0 , y)$ is a non-zero and non-decreasing function for $t\geq 0$.
		
		\item[(II)] No non-zero $y_0$ has the property that, for all $x\in\textnormal{\texttt{ri}}(\mathcal{A})$ and $y\in\textnormal{\texttt{ri}}(\mathcal{B})$, the half-line $\left\{ y + t y_0: t\geq 0 \right\}$ is contained in $\mathcal{B}$ and $F(x,y+ty_0)$ is a non-increasing function for $t\geq 0$.
	\end{itemize}
	If condition (I) is satisfied, then
	\begin{align} \label{eq:cond_I}
	\max_{x \in\mathcal{A}} \inf_{ y\in\mathcal{B}} F(x,y) 
	= \inf_{y\in\mathcal{B}} \sup_{x\in\mathcal{A}} F(x,y) < +\infty.
	\end{align}
	If condition (II) is satisfied, then
	\begin{align} \label{eq:cond_II}
	-\infty <
	\sup_{x \in\mathcal{A}} \inf_{ y\in\mathcal{B}} F(x,y) 
	= \min_{y\in\mathcal{B}} \sup_{x\in\mathcal{A}} F(x,y).
	\end{align}	
	If (I) and (II) are both satisfied, then $F$ has a saddle-point on $\mathcal{A}\times \mathcal{B}$.
\end{theo}


Fix an arbitrary $s \in \texttt{ri}\left( \mathbb{R}_{\geq 0} \right) = \mathbb{R}_{> 0}$.
We check that $\left( \mathcal{S}_{P,W}(\mathcal{H}), K_{R,P}(s, \cdot)\right)$ fulfills the three items in Definition \ref{defn:saddle}.
(a) The set $\mathcal{S}_{P,W}(\mathcal{H})$ is clearly convex.
(b) 
Since the map $\sigma \mapsto D_{1/(1+s)}(W\|\sigma|P)$ is convex (owing to Lieb's concavity theorem \cite{Lie73}) and lower semi-continuous on $\mathcal{L(H)}_+$ \cite[Corollary III.25]{MO14b}, 
by Eq.~\eqref{eq:K}, $\sigma \mapsto K_{R,P}(\alpha,\sigma)$ is also convex and lower semi-continuous on $\mathcal{S}_{P,W}(\mathcal{H})$.
(c) Due to the compactness of $\mathcal{S(H)}$, any accumulation point of $\mathcal{S}_{P,W}(\mathcal{H})$ that does not belong to $\mathcal{S}_{P,W}(\mathcal{H})$, say $\sigma_o$, satisfies $\sigma_o \in \mathcal{S(H)} \backslash \mathcal{S}_{P,W}(\mathcal{H})$.
By Eqs.~\eqref{eq:saddle23}, \eqref{eq:saddle22}, one finds $K_{R,P}(\alpha, \sigma_o) = +\infty$.


Next, fix an arbitrary $\sigma \in \texttt{ri}\left(  \mathcal{S}_{P,W}(\mathcal{H}) \right)$.
Owing to the convexity of $\mathcal{S}_{P,W}(\mathcal{H})$, it follows that $\texttt{ri}\left(  \mathcal{S}_{P,W}(\mathcal{H}) \right) $ $= \texttt{ri}\left(\texttt{cl}\left( \mathcal{S}_{P,W}(\mathcal{H})\right)\right)$ (see e.g.~\cite[Theorem 6.3]{Roc70}).
We first claim $\texttt{cl}\left( \mathcal{S}_{P,W}(\mathcal{H})\right) = \mathcal{S(H)}$.
To see this, observe that $\mathcal{S(H)}_{++} \subseteq  \mathcal{S}_{P,W}(\mathcal{H})$ since a full-rank density operator is not orthogonal with every $W_x$, $x\in\mathcal{X}$.
Hence, 
\begin{align}
\mathcal{S(H)}=
\texttt{cl}\left( \mathcal{S(H)}_{++} \right) 
\subseteq \texttt{cl} \left(  \mathcal{S}_{P,W}(\mathcal{H}) \right). \label{eq:saddle25}
\end{align}
On the other hand, the fact $\mathcal{S}_{P,W}(\mathcal{H}) \subseteq  \mathcal{S(H)}$ leads to
\begin{align}
\texttt{cl}\left( \mathcal{S}_{P,W}(\mathcal{H}) \right) \subseteq  
\texttt{cl}\left(\mathcal{S(H)}\right) 
= \mathcal{S(H)}. \label{eq:saddle26}
\end{align}
By Eqs.~\eqref{eq:saddle25} and \eqref{eq:saddle26}, we deduce that 
\begin{align}
\texttt{ri}\left(  \mathcal{S}_{P,W}(\mathcal{H}) \right)
= \texttt{ri}\left( \texttt{cl}\left( \mathcal{S}_{P,W}(\mathcal{H}) \right) \right)
= \texttt{ri}\left(  \mathcal{S}(\mathcal{H}) \right)
=  \mathcal{S(H)}_{++}, \label{eq:saddle24}
\end{align}
where the last equality in Eq.~\eqref{eq:saddle24} follows from \cite[Proposition 2.9]{Wei11}.
Hence, we obtain
\begin{align} \label{eq:saddle19}
\forall \sigma \in \texttt{ri}\left(  \mathcal{S}_{P,W}(\mathcal{H}) \right) \quad \text{and} \quad \forall x\in \mathcal{X},
\quad \sigma \gg W_x.
\end{align}
Now, we verify that $\left( \mathbb{R}_{\geq 0}, K_{R,P}(\cdot,\sigma)\right)$ satisfies the three items in Definition \ref{defn:saddle}.
(a) The set $\mathbb{R}_{\geq 0}$ is obviously convex.
(b) From Eqs.~\eqref{eq:saddle19} and the definition of the R\'enyi divergence , the map $s \mapsto D_{1/(1+s)}(W\|\sigma|P)$ is continuous on $\mathbb{R}_{\geq 0}$. 
Further, $s \mapsto s D_{1/(1+s)}(W\|\sigma|P)$ is concave on $\mathbb{R}_{\geq 0}$ \cite[Appendix B]{MO14b}.
By Eq.~\eqref{eq:K}, the map $s \mapsto K_{R,P}(s,\sigma)$ is concave and continuous on $\mathbb{R}_{\geq 0}$.
(c) Since $\mathbb{R}_{\geq 0}$ is closed, there is no accumulation point of $\mathbb{R}_{\geq 0}$ that does not belong to $\mathbb{R}_{\geq 0}$.

We are now in a position to prove item (i) of this Proposition.
Since the set $\mathcal{S}_{P,W}(\mathcal{H})$ is bounded, condition (II) is satisfied.
Equation \eqref{eq:cond_II} in Theorem \ref{theo:saddle} implies that
\begin{align}
- \infty < \sup_{ s\in\mathbb{R}_{\geq 0} } \inf_{\sigma \in \mathcal{S}_{P,W}(\mathcal{H})} K_{R,P} (s,\sigma)
= \min_{\sigma \in \mathcal{S}_{P,W}(\mathcal{H})} \sup_{ s\in\mathbb{R}_{\geq 0} }  K_{R,P} (s,\sigma). \label{eq:fact2}
\end{align}
Then Eqs.~\eqref{eq:fact1} and \eqref{eq:fact2} lead to the existence of a saddle-point of $K_{R,P}(\cdot,\cdot)$ on $\mathbb{R}_{\geq 0}\times \mathcal{S}_{P,W}(\mathcal{H})$.
Note that $K_{R,P}(s,\sigma) = F_{R,P}(1/(1+s),\sigma)$. 
We conclude the existence of a saddle-point of $F_{R,P}(\cdot,\cdot)$ on $(0,1]\times \mathcal{S}_{P,W}(\mathcal{H})$.
Hence, item (i) is proved.

%
%
%
%

We postpone the proof of the uniqueness of the optimizer to later and now show item (iii).
Given any $R\in(R_\infty, C   )$ and $P\in \mathscr{P}_R(\mathcal{X})$, one finds
\begin{align} \label{eq:saddle1}
\min_{\sigma \in \mathcal{S(H)}} \sup_{0<\alpha\leq 1} F_{R,P} (\alpha, \sigma) \in (0,+\infty).
\end{align}
If $\alpha^\star = 1$ and $\sigma^\star$ is a saddle point of $F_{R,P}(\cdot,\cdot)$, by Eq.~\eqref{eq:F} we deduce that $F_{R,P}(1, \sigma^\star) = 0$ for every possible $\sigma^\star$, which contradicts Eq.~\eqref{eq:saddle1}.
Hence, $\alpha^\star = 1$ is not a saddle point of $F_{R,P}(\cdot,\sigma^\star)$.

For any saddle-point $(\alpha^\star, \sigma^\star)$ of $F_{R,P}(\cdot,\cdot)$, it holds that
\begin{align} \label{eq:saddle3}
F_{R,P} (\alpha^\star, \sigma^\star) = \min_{\sigma\in\mathcal{S(H)}} F_{R,P} (\alpha^\star, \sigma)
= \frac{\alpha^\star - 1}{\alpha^\star} R + \frac{1-\alpha^\star}{\alpha^\star}  \min_{\sigma\in\mathcal{S(H)}}  D_{\alpha^\star} (W\|\sigma|P).
\end{align}
We claim the minimizer of Eq.~\eqref{eq:saddle3} must satisfy
\begin{align} \label{eq:saddle6}
\sigma^\star =  \frac{ 
	\left( \sum_{x\in\mathcal{X}} P(x) \frac{W_x^{\alpha^\star} }{ \Tr\left[ W_x^{\alpha^\star} (\sigma^\star)^{1-\alpha^\star} \right] }  \right)^{ \frac{1}{\alpha^\star}  } 	
}{  \Tr\left[  \left( \sum_{x\in\mathcal{X}} P(x) \frac{W_x^{\alpha^\star} }{ \Tr\left[ W_x^{\alpha^\star} (\sigma^\star)^{1-\alpha^\star} \right] }  \right)^{ \frac{1}{\alpha^\star}  } 	 \right] } 
=
\frac{\left( \sum_{x\in\mathcal{X}} P(x) W_x^{\alpha^\star} \mathrm{e}^{(1-\alpha^\star)D_{\alpha^\star}\left( W_x\| \sigma \right)} \right)^{ \frac{1}{\alpha^\star} }}{\Tr\left[ \left( \sum_{x\in\mathcal{X}} P(x) W_x^{\alpha^\star} \mathrm{e}^{(1-\alpha^\star)D_{\alpha^\star}\left( W_x\| \sigma \right)} \right)^{  \frac{1}{\alpha^\star}  }\right] }
\end{align}
for every $\alpha^\star \in (0,1)$.
Our approach closely follows from Hayashi and Tomamichel  \cite[Lemma 5]{HT14}.
For two density operators $\sigma, \omega \in\mathcal{S(H)}$ and a map $G:\mathcal{S(H)}\to \mathcal{L(H)}_\text{sa}$ (where $\mathcal{L(H)}_\text{sa}$ denotes the self-adjoint operators on $\mathcal{H}$), define the Fr\'echet derivative (see e.g.~\cite[Appendix C]{HT14}, \cite{HP14}\footnote{We note that the Fr\'echet derivative of functions involving matrices has other applications in quantum information theory; see e.g.~\cite{CH1,CH2,CH16RSPA}.})
\begin{align} \label{eq:Frechet}
\partial_\omega G(\sigma) := \mathsf{D} G (\sigma) [\omega - \sigma].
\end{align}
By letting
\begin{align}
g_\alpha (\sigma) := \sum_{x\in\mathcal{X}} P(x) \log \Tr \left[ W_x^\alpha \sigma^{1-\alpha} \right],
\end{align}
it follows that
\begin{align} \label{eq:target}
\sigma^\star  = \argmin_{\sigma\in\mathcal{S(H)}}  D_\alpha \left( W \| \sigma |P \right)
= \argmax_{\sigma\in\mathcal{S(H)}} g_\alpha(\sigma), \quad \forall \alpha \in (0,1).
\end{align}		
Since the map $\sigma\mapsto g_\alpha(\sigma)$ is strictly concave for every $\alpha\in(0,1)$ \cite{Lie73}, a sufficient and necessary condition for $\sigma$ to be an optimizer of Eq.~\eqref{eq:target} is $\partial_\omega g_\alpha (\sigma) = 0$ for all $\omega \in\mathcal{S(H)}$.
Direct calculation shows that
\begin{align} \label{eq:saddle9}
\partial_\omega g_\alpha (\sigma) 
= \Tr\left[  \sum_{x\in\mathcal{X}} P(x) \frac{W_x^\alpha}{\Tr\left[ W_x^\alpha \sigma^{1-\alpha} \right] } \partial_{\omega} \sigma^{1-\alpha}.
\right]
\end{align}
Next, we check that the fixed-points of the following map achieves the optimum:
\begin{align} \label{eq:opt2}
\sigma \mapsto  \frac{ 
	\left( \sum_{x\in\mathcal{X}} P(x) \frac{W_x^{\alpha} }{ \Tr\left[ W_x^{\alpha} \sigma^{1-\alpha} \right] }  \right)^{ \frac{1}{\alpha}  } 	
}{  \Tr\left[  \left( \sum_{x\in\mathcal{X}} P(x) \frac{W_x^{\alpha} }{ \Tr\left[ W_x^{\alpha} \sigma^{1-\alpha} \right] }  \right)^{ \frac{1}{\alpha}  } 	 \right] }.
\end{align}
Let 
\begin{align}
\chi_\alpha (\sigma) := \Tr\left[  \left( \sum_{x\in\mathcal{X}} P(x) \frac{W_x^{\alpha} }{ \Tr\left[ W_x^{\alpha} \sigma^{1-\alpha} \right] }  \right)^{ \frac{1}{\alpha}  } 	 \right] > 0, \quad \forall \alpha\in(0,1), \label{eq:chi}
\end{align}
and let	$\bar{\sigma}$ be a fix-point of the map in Eq.~\eqref{eq:opt2}.
Then, by Eqs.~\eqref{eq:opt2}, \eqref{eq:chi}, we have
\begin{align} \label{eq:saddle14}
\chi_\alpha (\bar{\sigma}) \cdot \bar{\sigma} = \left( \sum_{x\in\mathcal{X}} P(x) \frac{W_x^{\alpha} }{ \Tr\left[ W_x^{\alpha} \bar{\sigma}^{1-\alpha} \right] }  \right)^{ \frac{1}{\alpha}}.
\end{align}
Substituting Eq.~\eqref{eq:saddle14} into Eq.~\eqref{eq:saddle9} yields
\begin{align} \label{eq:saddle10}
\begin{split}
\partial_\omega g_\alpha (\bar{\sigma})
&= \Tr\left[ \chi_\alpha (\bar{\sigma})^\alpha \bar{\sigma}^\alpha  \partial_\omega \bar{\sigma}^{1-\alpha} \right]
= \Tr\left[ \chi_\alpha (\bar{\sigma})^\alpha \bar{\sigma}^\alpha ( 1-\alpha) \bar{\sigma}^{-\alpha} (\omega - \bar{\sigma})   \right] \\
&= (1-\alpha) \chi_\alpha (\bar{\sigma})^\alpha \Tr\left[ \omega - \bar{\sigma} \right] =
0.
\end{split}
\end{align}
By Brouwer's fixed-point theorem, the map in Eq.~\eqref{eq:opt2} is indeed the optimizer for Eq.~\eqref{eq:target}.	
Further, from  Eq.~\eqref{eq:saddle6}, it is clear that
\begin{align} \label{eq:saddle13}
\sigma^\star \gg W_x, \quad \forall x \in \texttt{supp}(P),
\end{align}
and thus item (iii) is proved.

Lastly, we show the uniqueness of the saddle-point.
Since  the map $\sigma\mapsto D_\alpha(W\|\sigma |P)$ is strictly concave \cite{Lie73}, the minimizer of Eq.~\eqref{eq:saddle1} is unique  for any $\alpha\in(0,1)$.	
Then, it remains to prove the uniqueness of the maximizer.
Let $\sigma^\star$ attain the minimum in Eq.~\eqref{eq:saddle1}.
By using the reparameterization again, we have
\begin{align}
K_{R,P}(s,\sigma^\star ) &= -sR + s   D_{\frac{1}{1+s}} \left( W \|\sigma^\star |P  \right) \\
\label{eq:saddle7}
&= -sR + s  \sum_{x\in\mathcal{X}} P(x) D_{\frac{1}{1+s}} \left( p_x \| q_x   \right),
\end{align}
where $p_x,q_x$ are the Nussbaum-Szko{\l}a distributions of $W_x$ and $\sigma^\star$.
The second-order partial derivative can be calculated as
\begin{align} \label{eq:saddle4}
\frac{\partial^2 K_{R,P}\left( s,\sigma^\star \right) }{ \partial  s^2 }  = -\frac{1}{(1+s)^3} \sum_{x\in\mathcal{X}} P(x)  \Var_{\hat{q}_{\frac{1}{1+s},x}} \left[ \log \frac{q_x}{p_x}  \right],
\end{align}
where
\begin{align}
\hat{q}_{t, x } (\omega) := \frac{ p_x(\omega)^{1-t}  q_x(\omega)^t }{\sum_{\omega \in \texttt{supp}(p_x) \cap \texttt{supp}(q_x)} p_x(\omega)^{1-t} q_x(\omega)^t  }, \quad \forall \omega \in \texttt{supp}(p_x) \cap \texttt{supp}(q_x), \; t\in[0,1].
\end{align}
Now, we assume the right-hand side of Eq.~\eqref{eq:saddle4} is zero, which is equivalent to
\begin{align} \label{eq:saddle11}
p_x (\omega) = c_x \cdot q_x (\omega), \quad \forall \omega \in \texttt{supp}(p_x) \cap \texttt{supp}(q_x)
\end{align}
for some constant $c_x>0$ and $x\in \texttt{supp}(P)$.
From Eq.~\eqref{eq:saddle13}, one finds $p_x \ll q_x$.
Summing the right-hand side of Eq.~\eqref{eq:saddle11} over $\omega  \in p_x^0$ yields
\begin{align} \label{eq:saddle12}
1 = c_x \cdot \Tr\left[ p_x^0 q_x \right], \quad \forall x\in \texttt{supp}(P).
\end{align}
By combining Eqs.~\eqref{eq:saddle11} and \eqref{eq:saddle12}, one can verify
\begin{align} \label{eq:saddle5}
\sup_{s\in\mathbb{R}_{> 0}} \left\{-sR+ s \sum_{x\in\mathcal{X}} P (x) D_{\frac{1}{1+s}}\left( p_x \|q_x \right) \right\} = \sup_{s\in\mathbb{R}_{> 0}} \left\{- s R -  s \sum_{x\in\mathcal{X}} P(x)\log \Tr\left[ p_x^0 q_x \right]   \right\} = 0,
\end{align}
where we rely the fact $R> R_\infty(W) \geq -\sum_{x\in\mathcal{X}} P(x) \log \Tr\left[ p_x^0 q_x \right]$ from Eq.~\eqref{eq:R_inf1}.
However, Eq.~\eqref{eq:saddle5} contradicts the assumption $P \in \mathscr{P}_R(\mathcal{X})$, which in turn implies that the right-hand side of Eq.~\eqref{eq:saddle4} is strictly negative.
Therefore, the map $s\to K_{R,P}(s,\sigma^\star)$ is strictly concave for all $s \in \mathbb{R}_{> 0}$ and thus the maximizer of Eq.~\eqref{eq:saddle1} is unique.
\qed

\subsection{Proof of Lemma \ref{lemm:hypothesis}} \label{proof:hypothesis}
Let $\mathbf{x}^n(m)$ be the codeword encoding the message $m \in \{1,\ldots,\exp\{nr\}\}$.
We define a binary hypothesis testing problem as:
\begin{align}
&\mathsf{H}_0: W_{\mathbf{x}^n (m)}^{\otimes n}; \\ 
&\mathsf{H}_1: \sigma^{n} := \bigotimes_{i=1}^n \sigma_i, 
\end{align}
where $\sigma^{ n} \in \mathcal{S}\left(\mathcal{H}^{\otimes n}\right)$ can be viewed as a dummy channel output.	
Since $\sum_{m=1}^M \beta\left( \Pi_{n,m}; \sigma^n \right) = 1$ for any POVM $\Pi_n = \{\Pi_{n,1},\ldots,\Pi_{n,\exp\{nr\}}  \}$, and $\beta\left( \Pi_{n,m} ; \sigma^{\otimes n} \right) \geq 0$ for every $m \in\mathcal{M}$, there must exist a message $m \in\mathcal{M}$ for any  code $\mathcal{C}_n$ such that $\beta\left( \Pi_{n,m} ; \sigma^{n} \right)\leq \exp\{-nr\} $. Let $\mathbf{x}^n := \mathbf{x}^n \left(m\right)$ be the codeword for that message $m$. Then 
\begin{align} \label{eq:sketch7}
\epsilon_{\max}\left(\mathcal{C}_n\right) \geq \epsilon_{m}\left( \mathcal{C}_n \right) = \alpha\left(\Pi_{n,m};  W_{\mathbf{x}^n}^{\otimes n}  \right)
\geq \widehat{\alpha}_{\exp\{-nr\}} \left( W_{\mathbf{x}^n}^{\otimes n} \| \sigma^{n} \right).
\end{align}
Since the above inequality \eqref{eq:sketch7} holds for every $\sigma^{n} \in \mathcal{S}\left(\mathcal{H}^{\otimes n}\right)$, it follows that
\begin{align}
\epsilon_{\max}\left(\mathcal{C}_n\right) \geq \max_{  \sigma \in \mathcal{S}(H) } \widehat{\alpha}_{\exp\{-nr\}} \left( W_{\mathbf{x}^n}^{\otimes n} \| \sigma^{\otimes n} \right).
\end{align}	
\qed

\subsection{Proof of Lemma \ref{lemm:regularity}} \label{proof:regularity}
This lemma closely follows from Altu\u{g} and Wagner's \cite[Lemma 9]{AW14}. However, the major difference is that we prove the claim using the expression $\phi_n$ as the error-exponent instead of the discrimination function: $\min\left\{\mathbb{D}\left(\tau\|\rho\right): \mathbb{D}\left(\tau\|\sigma\right) \leq R_n \right\}$. This expression is crucial to obtaining the sphere-packing bound in Theorem \ref{theo:refined} in the strong form of Gallager's expression

For convenience, we shorthand $r = R_n$.
From Lemma \ref{lemm:invariance}, it can be verified that
\begin{align}
E_0 (s) &:= -\frac{1+s}n \log \Tr \left[ \left( {p}^n\right)^{\frac1{1+s}} \left( {q}^n\right)^{\frac{s}{1+s}} \right] \\
&= - (1+s) {\Lambda}_{0,P_{\mathbf{x}^n}} \left( \frac{s}{1+s} \right), \label{eq:regularity8}
\end{align}
where Eq.~\eqref{eq:regularity8} follows from the definition of ${\Lambda}_{0,P_{\mathbf{x}^n}}$ in Eq.~\eqref{eq:FL0}.
Then, we rewrite the error-exponent function $\phi_n(r) $ by the Legendre-Fenchel transform of $E_0(s)$, i.e.,~
\begin{align}
{\phi}_n\left( r  \right) &=  \sup_{\alpha\in(0,1]} \left\{ \frac{\alpha-1}{\alpha} \left( r - \sum_{x\in\mathcal{X}} P_{\mathbf{x}^n}(x) D_\alpha\left(  {p}_x \|  {q}_x \right) \right) \right\}\\
&= \sup_{s\geq 0} \left\{ -sr + E_0(s)    \right\}. \label{eq:objective}
\end{align}

Direct calculation shows that
\begin{align} 
\frac{ \partial E_0(s) }{\partial s}
&= -{\Lambda}_{0,P_{\mathbf{x}^n}} \left( \frac{s}{1+s} \right) - \frac{1}{1+s} {\Lambda}'_{0,P_{\mathbf{x}^n}}  \left( \frac{s}{1+s} \right), \label{eq:regularity2} \\
\frac{ \partial^2 E_0(s) }{\partial s^2} &= -\frac{1}{(1+s)^3} {\Lambda}''_{0,P_{\mathbf{x}^n}} \left( \frac{s}{1+s} \right). \label{eq:regularity5}
\end{align}
Now assume the second-order derivative $\Lambda_{0,P_{\mathbf{x}^n}}''(t)$ in right-hand side of Eq.~\eqref{eq:regularity5} is zero for some $t\in[0,1]$. This is equivalent to
\begin{align} \label{eq:sharp33}
{p}_x (\omega) =  {q}_x (\omega) \cdot \mathrm{e}^{-\Lambda'_{0,x}(t)}, 
\quad \forall \omega \in p_x, 
\quad \forall x\in \texttt{supp}(P_{\mathbf{x}^n}).
\end{align}
Summing the right-hand side of Eq.~\eqref{eq:sharp33} over $\omega \in p_x$ gives
\begin{align} \label{eq:regularity123}
1 = \Tr\left[ p_x^0 q_x \right]  \mathrm{e}^{-\Lambda'_{0,x}(t)}.
\end{align}
Then, Eqs.~\eqref{eq:sharp33} and \eqref{eq:regularity123} imply that
\begin{align}
\phi_n ( r ) &= \sup_{0<\alpha\leq 1} \frac{\alpha-1}{\alpha} \left( r - \sum_{x\in\mathcal{X}} P_{\mathbf{x}^n} (x) D_\alpha\left( p_x \|q_x \right) \right) \\
&= \sup_{0<\alpha\leq 1} \frac{\alpha-1}{\alpha} \left(r+ \sum_{x\in\mathcal{X}} P_{\mathbf{x}^n} (x)\log \Tr\left[ p_x^0 q_x \right]    \right) = 0, \label{eq:sharp4}
\end{align}
where in Eq.~\eqref{eq:sharp4} we use the fact that $r > -\sum_{x\in\mathcal{X}} P_{\mathbf{x}^n} (x)\log \Tr\left[ p_x^0 q_x \right]$; see Eq.~\eqref{eq:R_inf1}.
However, from Lemma \ref{lemm:invariance} we know that $\phi_n(r) = E_\text{sp}(R) >0$, which leads to a contradiction.
Hence, we obtain
\begin{align} \label{eq:positive}
{\Lambda}''_{0,P_{\mathbf{x}^n}}(t) >  0, \quad \forall t\in[0,1],
\end{align}
and prove item (i).

From Eqs.~\eqref{eq:regularity5} and \eqref{eq:positive}, the objective function $-sr + E_0(s)$ in Eq.~\eqref{eq:objective} is strictly concave in $s$ for $s\in\mathbb{R}_{+}$. 
Further, by recalling that $\phi_n(r) = E_\text{sp}(R) >0$, $s=0$ will not be an optimum in Eq.~\eqref{eq:objective}.
We deduce that there exists a unique maximizer $s^\star \in \mathbb{R}_{> 0}$  such that
\begin{align}
r &= \left.\frac{ \partial E_0(s) }{\partial s}\right|_{s = s^\star}, \label{eq:regularity7} \\
{\phi}_n (r) &= E_0(s^\star) - s^\star \left.\frac{\partial E_0(s)}{\partial s}\right|_{s=s^\star}, \label{eq:regularity14}
\end{align}
if $r$ lies in the range:
\begin{align}
&-\frac1n\log \Tr\left[ (p^n)^0 q^n  \right] = \lim_{s\to+\infty} \frac{ \partial E_0(s) }{\partial s}
\leq 
r \leq \left. \frac{ \partial E_0(s) }{\partial s} \right|_{s=0} 
= \frac1n \mathbb{D}\left(  {p}^n\| {q}^n \right), \label{eq:regularity13}
\end{align}
where the boundary values $-\frac1n\log \Tr\left[ (p^n)^0 q^n  \right]$ and $\frac1n\mathbb{D}\left( {p}^n\| {q}^n\right)$ can be obtained from Eqs.~\eqref{eq:regularity2}, \eqref{eq:zero_derivative} and \eqref{eq:second_derivative}.
Substituting Eq.~\eqref{eq:regularity7} into \eqref{eq:regularity2} gives
\begin{align}
r = - {\Lambda}_{0,P_{\mathbf{x}^n}} \left( \frac{s^\star}{1+s^\star} \right) - \frac{1}{1+s^\star} {\Lambda}'_{0,P_{\mathbf{x}^n} }\left(  \frac{s^\star}{1+s^\star} \right). \label{eq:regularity10}
\end{align}
Further, Eqs.~\eqref{eq:objective}, \eqref{eq:regularity8}, \eqref{eq:regularity10} imply that
\begin{align}
{\phi}_n (r) &= - s^\star r + E_0(s^\star) \\
&= \frac{s^\star}{1+s^\star} {\Lambda}'_{0,P_{\mathbf{x}^n}}\left( \frac{s^\star}{1+s^\star} \right) - {\Lambda}_{0,P_{\mathbf{x}^n}} \left( \frac{s^\star}{1+s^\star} \right) . \label{eq:regularity9}
\end{align}
By comparing Eqs.~\eqref{eq:regularity10} and \eqref{eq:regularity9}, we obtain
\begin{align}
{\Lambda}'_{0,P_{\mathbf{x}^n}} \left( \frac{s^\star}{1+s^\star} \right) = {\phi}_n(r) - r \label{eq:regularity11}
\end{align}
which is exactly the optimum solution to the Fenchel-Legendre transform $\Lambda_{0,P_{\mathbf{x}^n} }^*(z)$ in Eq.~\eqref{eq:FL} with 
\begin{align}
t^\star &= \frac{s^\star}{1+s^\star}\in(0,1), \label{eq:opt_t}\\
z &=  {\phi}_n(r) - r.
\end{align}

From Eqs.~\eqref{eq:FL}, \eqref{eq:regularity11} and \eqref{eq:regularity9}, we conclude the item (i) of Lemma \ref{lemm:regularity}:
\begin{align}
\Lambda_{0,P_{\mathbf{x}^n} }^*\left(  {\phi}_n(r) - r \right)&= t^\star z - {\Lambda}_{0,P_{\mathbf{x}^n}}(t^\star) \\
&= \frac{s^\star}{1+s^\star}\left(  {\phi}_n(r) - r \right) - {\Lambda}_{0,P_{\mathbf{x}^n} }\left( \frac{s^\star}{1+s^\star} \right)\\
&= \frac{s^\star}{1+s^\star} {\Lambda}'_{0,P_{\mathbf{x}^n} } \left( \frac{s^\star}{1+s^\star} \right) - {\Lambda}_{0,P_{\mathbf{x}^n} }\left( \frac{s^\star}{1+s^\star} \right)\\
&=  {\phi}_n(r).
\end{align}

Item (ii) follows from item (i),
the symmetry $\Lambda_{0,x_i}(t) = \Lambda_{1,x_i} (1-t)$ and
$\Lambda_{0,x_i}'(t) = -\Lambda_{1,x_i}' (1-t)$,and Eq.~\eqref{eq:FL}.
$\Lambda_{1, P_{\mathbf{x}^n}  }^*\left( r - {\phi}(r)  \right) = r$.

For the item (iii), the positivity of ${\Lambda}''_{0,P_{\mathbf{x}^n} }(t)$, for $t\in[0,1]$, implies that the objective function $tz - \Lambda_{0,P_{\mathbf{x}^n} }(t)$ in Eq.~\eqref{eq:FL} is strictly concave in $t$ for $t\in [0,1]$.
Hence, by Eq.~\eqref{eq:opt_t}, the optimizer $t^\star \in(0,1)$ exists uniquely.
By recalling Eq.~\eqref{eq:regularity11}, we complete the claim in item (iii).
\qed



\end{document}